\documentclass[a4paper,twoside]{article}

\usepackage{epsfig}
\usepackage{subfigure}
\usepackage{calc}
\usepackage{amssymb}
\usepackage{amstext}
\usepackage{amsmath}
\usepackage{stmaryrd}
\usepackage{multicol}
\usepackage{pslatex}
\usepackage{apalike}
\usepackage{fancyhdr}
\usepackage{scitepress}
\usepackage[small]{caption}

\usepackage{version}

\includeversion{LONG}
\excludeversion{SHORT}

\usepackage[latin1]{inputenc}
\usepackage{url}
\usepackage{tikz}
\usepackage{graphicx,fancyvrb}
\usepackage{algorithmic}
\usepackage{algorithm}
\usepackage{latexsym}
\usepackage{amsfonts}
\usepackage{graphicx}
\usetikzlibrary{arrows,shapes}
\newcommand{\myunit}{1 cm}
\tikzset{node style sp/.style={draw,circle,minimum size=\myunit}}
\tikzset{node style el/.style={draw,ellipse,minimum size=\myunit}}

\newtheorem{example}{Example}
\newtheorem{theorem}{Theorem}
\newtheorem{lemma}{Lemma}

\newenvironment{proof}{{\bfseries Proof }}{\hfill$\square$}

\newcounter{thealgorithm}
\newenvironment{myalgorithm}[1]{\refstepcounter{thealgorithm}\noindent \rule{\linewidth}{0.03cm}
\textbf{Algorithm \arabic{thealgorithm}} {\scshape #1}\vspace{-3mm}\newline
\rule{\linewidth}{0.03cm}}{\vspace{-3mm}\rule{\linewidth}{0.03cm}}

\newcommand{\secref}[1]{\S\ref{#1}}

\begin{document}
\onecolumn \maketitle \normalsize \vfill

\section{Introduction}


\begin{LONG}
\subsection{Context and motivation}
\noindent
\end{LONG}
\begin{SHORT}
{\bf Context and motivation.}
\end{SHORT}
Similar to what happens between humans in the real world, in open
multi-agent systems~\cite{David01} distributed over the Internet,
such as online social networks or wiki technologies, agents often form
coalitions by agreeing to act as a whole in order to achieve certain
common goals. For instance, agents may wish to collaborate in order to
jointly create and use a group cryptographic key for ensuring the
confidentiality and/or integrity of information shared within the group,
e.g.~\cite{RafaeliHutchison03}, or to partake in a mix network or some
other anonymous remailer to achieve unobservability of communications,
e.g.~\cite{Chaum81}, or to create secret interest groups within online
social networks, e.g.~\cite{SorniottiMolva10}. However, agent coalitions
are not always a desirable feature of a system, as malicious or corrupt
agents may collaborate in order to subvert or attack the system. For
instance, such agents may collaborate to attack the information in
transit over different channels in a web service architecture or in a
distributed wired and/or wireless computer network,
e.g.~\cite{Wiehler04}, or they might forge and spread false information
within the system, e.g.~\cite{RePEc:jas:jasssj:2006-13-4}.

In order to be able to rigorously formalize and reason about such
positive and negative properties of agent coalitions, and thereby allow
for the prevention or, at least, the identification of the entailed
vulnerabilities, a number of different formal approaches have been
recently proposed,
\begin{LONG}
such
as~\cite{quantified2,quantified,atl,strategy,cl-pc2,pauly,cl-pc3,cl-pc,quantified3}.
\end{LONG}
\begin{SHORT}
such
as~\cite{quantified,atl,strategy,cl-pc2,pauly,cl-pc3,cl-pc}.
\end{SHORT}



 


In this paper, we consider the problem of \emph{hidden coalitions}: a
coalition is hidden in a system when its existence and the purposes it
aims to achieve are not known to the system. Hidden coalitions carry out
\emph{underhand attacks}, a term that we borrow from military terminology. These attacks are particularly subtle since
the agents that perform them are not outsiders but rather members of the
system whose security properties are posed under threat. Moreover, the
mere suspect that a group of individuals act as a whole is 
typically
insufficient to come to a decision about their permanence as members of
the system; this, of course, depends also on the nature of the system
and the information it contains, since in the presence of highly
security-sensitive information, systems may anyway opt for the exclusion
of all suspected agents. However, in general, systems, and even more so
open ones, will want to adopt a less restrictive policy, excluding only
those agents whose malice has indeed been proved. Therefore, the defense
against underhand attacks by hidden coalitions is a fundamental but
complex matter.



Problems of a similar kind have been studied, for instance, in Game
Theory~\cite{game,pauly3} in relation to the nature of collaboration and
competition, and from the viewpoint of modeling group formation under
the constraints of possible given goals. However, underhand attacks by
hidden coalitions pose security problems that cannot be dealt with such
traditional means. 
\begin{SHORT}
Nor can they solved by a simple, monotonic, approach based on Coalition
Logic(s) such as~\cite{quantified,cl-pc2,pauly,cl-pc}.
\end{SHORT}
\begin{LONG}
Nor can they solved by a simple, monotonic, approach based on Coalition
Logic(s)~\cite{quantified,cl-pc2,pauly,cl-pc}, which is currently one of
the most successful formalisms for reasoning about coalitions.
\end{LONG}





To illustrate all this further, consider the following concrete example
from an online social network such as Facebook, where abuse, misuse or
compromise of an account can be reported to the system administration.
In particular, a group of agents (in this case, Facebook users) can report a fake profile:
\begin{quote}\footnotesize
You can report a profile that violates Facebook's Statement of Rights
and Responsibilities by clicking the ``Report/Block this Person'' link
in the bottom left column of the profile, selecting ``Fake profile'' as
the reason, and adding the appropriate information. [...] 
\emph{(Excerpt from} \url{http://www.facebook.com/help/?search=fake}\emph{)}
\end{quote}
The administrator of the system gives an ultimatum to the agent that
uses the reported profile and then may, 
eventually, close it. 
An underhand coalition can exploit this report mechanism to attack an
agent who possesses a ``lawful'' original profile: 
at first they create a
fake profile with personal information and photos of the agent under
attack, and then they become friends of her. After that, they report the
original profile so that the administrator closes it. The report is a
lawful action, and by creating the new profile and having a big enough
number of agents who report the same profile no suspicion about the hidden coalition is raised, so that the attack succeeds.

\begin{LONG}
\subsection{Contributions}
\noindent\end{LONG}
\begin{SHORT}
{\bf Contributions.}
\end{SHORT}
A formalism to define and reason about such hidden coalitions is thus
needed. Indeed, Coalition Logic allows one to define coalitions that are
explicit (i.e.~not hidden) and is characterized by monotonic permissions
to act in groups and individually. What is missing, however, is the
notion of hidden coalition and a method to block the underhand attacks
such coalitions carry out. The idea underlying our approach is to
circumscribe the problem in algebraic terms, by defining a system that
can be represented by a coalition logic, and then activate a
non-monotonic control on the system itself to block the underhand
attacks that hidden coalitions are attempting to carry out.

More specifically, we consider multi-agent systems whose security
properties depend on the values of sets of logical formulas of
propositional logic, which we call the \emph{critical} (or \emph{security}) \emph{formulas} of the
systems: for concreteness, we say that a system is \emph{secure} if all
the critical formulas are false, and is thus insecure if one or more
critical formula is true. (Of course, we could also invert the
definition and consider a system secure when all critical formulas are
true.) The system agents control the critical formulas in that they
control the propositional variables that formulas are built from: we
assume that every variable of the system is controlled by an agent,
where the variables controlled by an agent are controlled just by that
agent without interference by any other agent. The actions performed by
each agent consist thus in changing some of the truth values of the
variables assigned to that agent, which means that the values of the
critical formulas can change due to actions performed by the agents,
including in particular malicious insider agents who form hidden
coalitions to attack the system by making critical formulas become true.
Returning to the Facebook example, this is exactly what happens when
agents report the original profile as fake by setting the flag (clicking
on the link).\footnote{In this paper, we do not consider how the
administrator decides to close the profile, nor do we consider in detail
the non-monotonic aspects of how agents enter/exit/are banned from a
system or enter/exit a hidden coalition, or how members of a hidden
coalition synchronize/organize their actions. All this will be subject
of future work.}

At each instant of time, agents ask the system to carry out the actions
they wish to perform, i.e.~changing the truth value of the variables
they control, and the system has to decide whether to allow such
actions, but without knowing of the existence of possible hidden
coalitions and thus at the risk of the system becoming insecure. To
block such attacks, we formalize here a \emph{deterministic blocking
method}, implemented by a greedy algorithm, which blocks the actions of
potentially dangerous agents. We prove that this method is sound and
complete, in that it does not allow a system to go in an insecure state
when it starts from a secure state and it ensures that every secure
state can be reached from any secure state. However, this algorithm is
not optimal as it does not block the smallest set of potentially
dangerous agents.
\begin{SHORT}
We thus introduce also a \emph{non-deterministic blocking method}, which
we obtain by extending the deterministic method with an oracle to
determine the minimum set of agents to block so to ensure the security
of the system. We show that the soundness and completeness result
extends to this non-deterministic method.
\end{SHORT}
\begin{LONG}

We thus introduce also a \emph{non-deterministic blocking method}, which
we obtain by extending the deterministic method with an oracle to
determine the minimum set of agents to block so to ensure the security
of the system. We show that the soundness and completeness result
extends to this non-deterministic method as well.
\end{LONG}

We also calculate the computational cost of our two blocking methods.
This computational analysis is completed by determining upper bound
results for the problem of finding a set of agents to be blocked so to
prevent system transitions into insecure states, and the problem of
finding an optimal set of agents satisfying the above condition.

\begin{LONG}
\subsection{Organization of the paper} 
\end{LONG}
\begin{SHORT}
{\bf Organization.} 
\end{SHORT}
In \secref{system}, we introduce our approach to the
problem of blocking underhand attacks by hidden coalitions.
\secref{method} and \secref{non-determ} respectively introduce our
deterministic and non-deterministic blocking methods, giving concrete
examples for their application. In \secref{properties}, we study the
computational aspects of these two methods, calculating in particular
their computational cost, and show that they are both sound and
complete. Finally, in \secref{concl}, we summarize our main results,
discuss related work and sketch future work.
\begin{SHORT}
Proofs of the formal results are given in~\cite{CKV10-arxiv}.
\end{SHORT}
\begin{LONG}
\bigskip 
\end{LONG}

\section{An approach to the problem of blocking underhand attacks}
\label{system}
\noindent We introduce our approach to the problem of blocking underhand
attacks.
\begin{SHORT}
We also recall some basic notions and, in particular, the relevant
notions of the Coalition Logic of Propositional Control
CL-PC~\cite{cl-pc2,cl-pc}.
\end{SHORT}
\begin{LONG}
We also recall some basic notions and, in particular, the relevant
notions of the Coalition Logic of Propositional Control
CL-PC~\cite{cl-pc2,cl-pc}, which provides a starting point for our
approach.
\end{LONG}


%

\subsection{Syntax}

We consider multi-agent systems $\mathcal{S}$ that are described by a
set of \emph{critical} (or \emph{security}) \emph{formulas} $\Phi$ and
by a temporal sequence $\sigma: T \rightarrow \Theta(\Phi)$, with $T$
the temporal axis and $\Theta(\Phi)$ the propositional assignment in the
set of formulas. In this work, we focus only on the formulas in $\Phi$,
which represent the security-critical characteristics of a system (which
depend on the application and which we thus do not describe further
here, as our approach is independent of the particular application). We
say that a system is \emph{secure} if all the critical formulas are
false, and it becomes insecure if one or more $\phi \in \Phi$ becomes
true.\begin{LONG}\footnote{As we remarked above, we could also invert
this and call a system secure when all critical formulas are true; we
would then just need to modify our methods accordingly.}\end{LONG}
Hence, the \emph{state} of a system is defined by the value of the
propositional variables that occur in the critical formulas of $\Phi$.

The \emph{agents} of a system $\mathcal{S}$ control the set $\Phi$ and
hence the state of $\mathcal{S}$. We require that there is no formula in
our systems that cannot change its truth value. Moreover, the
distribution of the variables to the agents should be such that one
formula cannot be controlled by one single agent, but rather different
agents control one formula, and every formula is controlled by some
agents. In particular, for a set $\mathit{Ag}$ of system agents:
\begin{itemize} 
\item every variable of the system is controlled by an agent $a \in
\mathit{Ag}$, and
\item the variables controlled by an agent are controlled just by that
agent without interference by any other agent.
\end{itemize}
The actions performed by each agent $a \in \mathit{Ag}$ are thus the
changing of the truth values of the variables assigned to $a$.
The agents we consider are \emph{intelligent agents} in the sense of~\cite{ag2,ag3}: they are autonomous,
have social ability, reactivity and pro-activeness, and have mental
attitudes, i.e.~states that classify the relation of the agents and the cognitive environment. In our approach, we consider intelligent agents
but do not make specific assumptions about their mental attitudes,
except for their collaborative attitudes that constitute a threat to
(the security of) the system.\footnote{An extension of the work presented here with a detailed formalization of the mental and collaborative attitudes of the agents will be subject of future work.}

In Game Theory~\cite{strategy}, strategies are often associated with a
preference relation for each agent that indicates which output the agent
is going to select in presence of alternatives. In our approach, agents
change the value of ``their'' variables according to their strategies
and create coalitions with other agents so to be more expressive: by
collaborating, agents can change the values of different variables and
thus, ultimately, of the critical formulas that comprise such variables.
The novelty in this work is that we don't deal just with coalitions that
are known by the system but also with hidden coalitions, whose existence
and purposes are unknown by the system.

Let us now formalize the language of our approach. Following CL-PC,
given a set $\mathit{Ag}$ of agents, a set $\mathit{Vars}$ of
propositional variables, the usual operators $\neg$ and $\vee$ of
classic propositional logic, and the \emph{cooperation mode} $\Diamond$,
we consider formulas built using the following grammar:
\begin{displaymath}
\phi ::= \top \mid p \mid\neg{\phi} \mid \phi \vee \phi \mid \Diamond_{C}\phi 
\end{displaymath}%
where $p\in \mathit{Vars}$, $C \subseteq Ag$, and $\Diamond_{C}\phi$ is
a \emph{cooperation formula}. Slightly abusing notation, we denote with
$\mathit{Vars}(\phi)$ the set of propositional variables that occur in
$\phi$ and with $\mathit{Ag}(\phi)$ the agents that control the
variables in $\mathit{Vars}(\phi)$.  $\Diamond_{C}\phi$
expresses that the coalition $C$ has the contingent ability to achieve
$\phi$; this means that the members of $C$ control some variables of $\phi$ and have choices for $\phi$ such that if they make
these choices and nothing else changes, then $\phi$ will be
true. 


\subsection{Semantics}
\label{semantics}

\begin{LONG}
A model is a tuple 
\begin{displaymath}
\mathcal{M} = \langle \mathit{Ag}, \mathit{Vars},
\mathit{Vars}|_1,...,\mathit{Vars}|_n,\theta\rangle\,,
\end{displaymath} 
where:
\begin{itemize}
\item $\mathit{Ag} = \{1,...,n\}$ is a finite, non-empty set of agents;
\item $\mathit{Vars} = \{p,q,...\}$ is a finite, non-empty set of
propositional variables;
\item $\mathit{Vars}|_1,...,\mathit{Vars}|_n$ is a partition of $\mathit{Vars}$ among the members
of $\mathit{Ag}$, with the intended interpretation that $\mathit{Vars}|_i$ is the
subset of $\mathit{Vars}$ representing those variables under the control
of agent $i\in \mathit{Ag}$;
\item $\theta: \mathit{Vars} \rightarrow \{\top,\bot\}$ is a
propositional valuation function that determines the truth value of each
propositional variable.
\end{itemize}
Since $\mathit{Vars}|_1,...,\mathit{Vars}|_n$ is a partition of
$\mathit{Vars}$, we have:
\begin{enumerate}
\item $\mathit{Vars} = \mathit{Vars}|_1 \cup ...\cup \mathit{Vars}|_n$
i.e.~every variable is controlled by some agent;
\item $\mathit{Vars}|_i \cap \mathit{Vars}|_j = \emptyset$ for $i \not =
j \in \mathit{Ag}$, i.e.~no variable is controlled by more than one
agent.
\end{enumerate}
\end{LONG}
\begin{SHORT}
A model is a tuple $\mathcal{M} = \langle \mathit{Ag}, \mathit{Vars},	\mathit{Vars}|_1,...,\mathit{Vars}|_n,$ $\theta\rangle$, where:
$\mathit{Ag} = \{1,...,n\}$ is a finite, non-empty set of agents;
$\mathit{Vars} = \{p,q,...\}$ is a finite, non-empty set of
propositional variables;
$\mathit{Vars}|_1,...,\mathit{Vars}|_n$ is a partition of $\mathit{Vars}$ among the members
of $\mathit{Ag}$, with the intended interpretation that $\mathit{Vars}|_i$ is the
subset of $\mathit{Vars}$ representing those variables under the control
of agent $i\in \mathit{Ag}$;
$\theta: \mathit{Vars} \rightarrow \{\top,\bot\}$ is a
propositional valuation function that determines the truth value of each
propositional variable.

Since $\mathit{Vars}|_1,...,\mathit{Vars}|_n$ is a partition of
$\mathit{Vars}$, we have:
$\mathit{Vars} = \mathit{Vars}|_1 \cup ...\cup \mathit{Vars}|_n$
i.e.~every variable is controlled by some agent; and
$\mathit{Vars}|_i \cap \mathit{Vars}|_j = \emptyset$ for $i \not =
j \in \mathit{Ag}$, i.e.~no variable is controlled by more than one
agent.
\end{SHORT}%
We denote with $\mathit{Vars}|_{C}$ the variables controlled by the
agents that are part of the coalition $C \subseteq \mathit{Ag}$.
Given a model $\mathcal{M} = \langle Ag, Vars,
\mathit{Vars}|_1,...,\mathit{Vars}|_n,\theta\rangle$ and a coalition
$C$, a $C$-valuation function is
$\theta_{C}:\mathit{Vars}|_{C}\rightarrow \{\top, \bot\}$.
Valuations $\theta$ extend from variables to formulas in the usual way
and for a model $\mathcal{M} = \langle Ag, Vars,
\mathit{Vars}|_1,...,\mathit{Vars}|_n,\theta\rangle$
we write $\mathcal{M} \models \phi$ if $\theta(\phi) = \top$. We write
$\models \phi$ if $\mathcal{M} \models \phi$ for all $\mathcal{M}$.

\subsection{Secure and insecure systems}

All the semantic notions introduced above actually depend on the current
time, and we will thus decorate them with a superscript $\cdot^{S_t}$
denoting the system state at time $t$, e.g.~$\theta^{S_t}$ and
$\models^{S_t}$. Time is discrete and natural, and is defined with a non
empty set of time points $T$ and a transitive and irreflexive relation
$\prec$ such that $t\prec u$ means that $t$ comes before $u$ for $t, u
\in T$. In our case, since $t, t+1 \in T$ it follows naturally that
$t\prec t+1$.


The passing of time is regulated by a general \emph{clock}, which
ensures that the system can execute a definite number of actions in an
instant of time: at every clock of time, the system changes its state,
which is thus defined by the actions that the system executes. Even if
there are no actions to execute, the system changes its state from $S_t$
to $S_{t+1}$, which in this case are equal.

We assume that each system $\mathcal{S}$ starts, at time $t_0$, from a
secure state $S_0$, i.e.~a state in which all the critical formulas of
$\Phi$ are false, so that none of the features of the system is
violated. In general:
\begin{center}
$\mathcal{S}$ is \emph{secure at a state} $S_t$ iff $\not\models^{S_t}
\phi$ for all $\phi\in\Phi$
\end{center}	
\begin{SHORT}
and $\mathcal{S}$ is \emph{secure} iff it is secure at all $S_t$.
\end{SHORT}
\begin{LONG}
and
\begin{center}
$\mathcal{S}$ is \emph{secure} iff $\not\models^{S_t}
\phi$ for all $\phi\in\Phi$ and all $S_t$
\end{center}	
\end{LONG}

At time $t$, the system is in state $S_t$ and goes to state $S_{t+1}$
and executes all the actions of the agents that want to change the value
of their variables. Denoting with $\Gamma_{t+1}$ the set of actions that
the agents want to execute at the time instant $t$, we can write
$$
S_t \stackrel{\Gamma_{t+1}}{\longmapsto} S_{t+1}\,.
$$
and the aim of our approach is to guarantee that each reachable state
$S_{t+1}$ is secure, where the differences between $S_t$ and $S_{t+1}$ are
in their respective $\Theta$.


Since a coalition can change the value of the variables it controls, it
can attempt to change the value of a critical formula to true; formally,
for a coalition $C$ and a formula $\phi$ if $\Diamond_{C}\phi$ is true
then it means that $C$ can make $\phi$ true and thus the system
insecure, which we can write by negating the above definition or
alternatively, and basically equivalently, as:
\begin{center}
$\mathcal{S}$ is \emph{insecure} at a state $S_t$ iff $\models^{S_t}
\Diamond_{C}\phi$ for some $C \subseteq Ag$ and some $\phi\in\Phi$ 
\end{center}




To help the control of the system (but without loss of generality), we
can create a \emph{filter} for the actions that imposes a limit on the
number of the actions that can be executed in an instant of time. This
can decrease the performance of the system, so we need a trade-off
between control and performance.

\section{A deterministic blocking method}
\label{method}


\noindent Our aim is to introduce a method that guarantees the security
of the system, which amounts to blocking the actions of hidden
coalitions. Indeed, in the case of ``normal'' coalitions, the property
$\Diamond_{C}\phi$ allows us to list the actions of the agents in $C$,
while if the coalition is hidden then we cannot block any action as we
cannot directly identify the participants of a coalition we do not even
know to exist. Since the actions of participants of hidden coalitions
are not predictable, we cannot oppose these coalitions using $\Diamond$,
so we introduce a method that disregards the existence of this property.

Our (main) method for the protection of the system is a \emph{blocking
method} based on the \emph{greedy} Algorithm~\ref{blocking}: the agents
make a request to the system for the actions $\Gamma_{t+1}$ they wish to
execute at time $t$, and the system then simulates (via a method
$\textit{Simulate}$ we assume to exist) the actions in order to control
whether the system after the execution of the actions is still secure or
not. The simulation says if the system can proceed with the execution of
the actions or not, in which case it is given a list of the formulas
$\Phi'$ that became true along with the set of agents $\mathcal{A}'$
that made them become true.



\begin{figure}[t]
\begin{myalgorithm}{A greedy, deterministic blocking method}
\label{blocking}
\begin{algorithmic}[1]
\STATE  $\textit{Simulate}(\Gamma_{t+1}) = [\Phi' \mathcal{A}']$;
\WHILE{($\Phi'\neq \emptyset$)}
	\STATE Create the matrix with $\Phi'$ and $\mathcal{A}';$
	\STATE $\forall a_i\in \mathcal{A}':a_i\rightarrow c_i,  c_i=count(\phi_i);$
	\STATE $Quicksort(c_1,..c_k)=(c_x,...);$
	\STATE $\mathcal{B}=\mathcal{B}\cup a_x;$ \COMMENT{where $a_x$ is the agent associated to $c_x$, that is the maximum counter of the marked cells}
	\STATE $\textit{Simulate}(\Gamma_{t+1} \backslash \Gamma|_{a_x})
 = [\Phi'$ $\mathcal{A}']$;
\ENDWHILE
\end{algorithmic}
\end{myalgorithm}
\end{figure}
 

If the simulation says that the system can go in an insecure state, the
blocking method constructs a matrix: in every column of the matrix there
is one of the agents given by the simulation and in every row there is
one of the formulas that became true during the simulation. We mark each
cell that has as coordinates the agent that has variables in that
formula, and then we eliminate the column that has more marked
cells.\footnote{It would be more efficient to consider only
the variables of the formulas that become true, but if we take only
these variables, we cannot prevent long-term strategies of hidden
coalitions, consisting in the progressive reduction of the number of
steps needed for making a security formula true. An optimization of the
choice of variables to be considered in order to reduce the
effectiveness of such long-term strategies will be
subject of future work.}
The corresponding agent is not eliminated, rather he is just blocked and
his actions are not executed (by subtracting $\Gamma|_{a_x}$): the
``dangerous'' agents found in this way are put in a set $\mathcal{B}$ of
blocked agents. The simulation is called again and so on, until the
output of the simulation is an empty set of formulas, which means that
by executing the remaining actions the system does not go in an insecure
state. It is important to note that this method does not prevent the
creation of hidden coalitions but can guarantee the system security from
the attacks made by these coalitions.

The most important property of Algorithm~\ref{blocking} is that it never
brings the system in an insecure state, as it blocks the actions of agents
that can make the system insecure. We do not commit to a specific way that
the blocking is actually done, as it depends on the particular observed
systems and on the particular goals. For instance:
\begin{itemize}
\item Block the agent from changing the value of his variables until a
precise instant of time. During this period, his variables are left
unchanged or are controlled by the superuser/system administrator.
\item Block the agent for an interval of time, which
can be a default value or can be chosen in a random way, e.g.~so that a
hidden coalition doesn't know when the agent can be active and thus
cannot organize another attack.
\item Block the agent and remove his actions for that instant of time.
At the next instant, the agent has the possibility to ask for his
actions to be executed.
\item Leave the variables unchanged, 
without making known to the agent if the value of the variables has been
changed or not. This method can be improved by blocking the agent if he
attempts to change the truth value of those variables again.
\end{itemize}
Other, more complex, blocking strategies can of course be given, e.g.~by combining some the above. 

Note also that, depending on the system considered, it could be that not
all the requests for execution can be satisfied: \emph{the maximum
number $n$ of actions that can be executed in an instant of time} can be
chosen in different ways, with respect to the characteristics of the
system. Here, we choose $n$ to be the cardinality $|\Phi|$ of the
critical formulas. The order used for taking these actions and executing
them respects a FIFO queue, so the first $n$ actions are executed.

\begin{example}\label{example1}
As a concrete example of the application of the blocking method, consider a system $\mathcal{S}$ defined by the critical formulas
\begin{eqnarray*}
\phi_1 & = & v_1 \wedge v_2 \wedge (\neg v_3 \vee v_5 \vee \neg v_4)
\\
\phi_2 & = & (\neg v_5 \vee \neg v_3) \wedge \neg v_6
\\
\phi_3 & = & v_7 \wedge (\neg v_8 \vee \neg v_6) \\
\phi_4 & = & (v_8 \vee v_5 \vee \neg v_9) \wedge v_2 \wedge v_1
\end{eqnarray*}
so that number of the action to be executed in an instant of time is
$n=4$ (the cardinality of the set of critical formulas that define the
system), and let $\mathit{Ag} = \{a_1,a_2, a_3, a_4, a_5\}$ and $At =
\{v_1,...,v_9 \}$. Further, consider the following distribution of the
variables to the agents:
\begin{SHORT}
$$
\begin{array}{lll}
a_1 = \{v_1, v_7, v_8\} 
& \quad
a_2 = \{v_3\}
& \quad
a_3 = \{v_2, v_6\} \\
a_4 = \{v_4, v_5\} 
& \quad
a_5 = \{v_9\}
\end{array}
$$
\end{SHORT}%
\begin{LONG}
\begin{eqnarray*}
a_1 & = & \{v_1, v_7, v_8\} \\
a_2 & = & \{v_3\} \\
a_3 & = & \{v_2, v_6\} \\
a_4 & = & \{v_4, v_5\} \\
a_5 & = & \{v_9\}
\end{eqnarray*}
\end{LONG}%
Let us assume that the state $S_t$ at time $t$ is
\begin{displaymath}
\begin{array}{l}
\theta^{S_t}(v_1) = \theta^{S_t}(v_5) = \bot \\
\theta^{S_t}(v_2) = \theta^{S_t}(v_3) = \theta^{S_t}(v_4) = \theta^{S_t}(v_6) = \theta^{S_t}(v_7) \\
\phantom{\theta^{S_t}(v_2) } = \theta^{S_t}(v_8) = \theta^{S_t}(v_9)  =  \top
\end{array}
\end{displaymath}
and that we have the following actions $\Gamma_{t+1}$ to be executed
at time $t$ in the FIFO queue:
\begin{SHORT}
$$
\begin{array}{lll}
\theta^{S_{t+1}}(v_1) \mapsfrom \top\,, & \quad
\theta^{S_{t+1}}(v_3) \mapsfrom \bot\,, \\
\theta^{S_{t+1}}(v_4) \mapsfrom \bot\,, & \quad
\theta^{S_{t+1}}(v_6) \mapsfrom \bot\,, & \quad
\cdots
\end{array}
$$
\end{SHORT}%
\begin{LONG}
\begin{eqnarray*}
\theta^{S_{t+1}}(v_1) & \mapsfrom & \top\,, \\
\theta^{S_{t+1}}(v_3) & \mapsfrom & \bot\,, \\
\theta^{S_{t+1}}(v_4) & \mapsfrom & \bot\,, \\
\theta^{S_{t+1}}(v_6) & \mapsfrom & \bot\,, \\
\cdots
\end{eqnarray*}
\end{LONG}%
That is, $v_1$ should be set to $\top$ at state $S_{t+1}$, and so on.
%
The algorithm simulates the first $n=4$ actions, so that $\Phi' =
\{\phi_1, \phi_2, \phi_3,\phi_4\}$ and $\mathcal{A}' =~\{a_1, a_2, a_3,
a_4 \}$, and the matrix of Table~\ref{table:tab1} is constructed, which the algorithm sorts by the highest counter to produce the matrix in
Table~\ref{table:tab2}. $a_3$ is thus put into $\mathcal{B}$. The
simulation takes place again, taking into account that we have blocked
the value of the variables controlled by $a_3$ at the truth value of the
instant of time $t$. The simulation gives as result the set $\Phi' =
\{\phi_1, \phi_4\}$ e $\mathcal{A}' = \{\phi_1, \phi_2, \phi_4\}$.
The matrix of Table~\ref{table:tab3} is created, which is already ordered (so the sorting will return the same matrix).
So, we put in $\mathcal{B}$ the agent $a_1$, block its actions and make
the simulation with the remaining actions. This simulation gives
$\Phi'=\emptyset$, and thus the remaining actions can be executed
without any risk for the system $\mathcal{S}$.
\end{example}

\begin{table}[t]
\caption{Matrix constructed by the blocking algorithm for Example~\ref{example1}.}
\begin{center}\footnotesize
\begin{tabular}{p{0.5cm}||*{4}{c}}
                            & $a_1$ & $a_2$ & $a_3$ & $a_4$ \\
\hline
\hline
\bfseries $\phi_1$          & X & X & X & X \\
\hline
\bfseries $\phi_2$          &   & X & X & X \\
\hline
\bfseries $\phi_3$          & X &   & X &   \\
\hline
\bfseries $\phi_4$         & X &   & X & X \\
\end{tabular}
\end{center}
\label{table:tab1}
\end{table}

\begin{table}[t]
\caption{Matrix of Table~\ref{table:tab1} sorted in a decreasing order of counters.}
\begin{center}\footnotesize
\begin{tabular}{p{0.5cm}||*{4}{c}}
                            & $a_3$ & $a_1$ & $a_4$ & $a_2$ \\
\hline
\hline
\bfseries $\phi_1$          & X & X & X  & X\\
\hline
\bfseries $\phi_2$          & X &   & X  & X \\
\hline
\bfseries $\phi_3$          & X & X &   &   \\
\hline
\bfseries $\phi_4$         & X & X  & X &   \\
\end{tabular}
\end{center}
\label{table:tab2}
\end{table}

\begin{table}[t]
\caption{Matrix of Table~\ref{table:tab2} after the block of agent $a_3$.}
\begin{center}\footnotesize
\begin{tabular}{p{0.5cm}||*{3}{c}}
                            & $a_1$ & $a_2$ & $a_4$ \\
\hline
\hline
\bfseries $\phi_1$          & X & X  &   \\
\hline
\bfseries $\phi_4$         & X &   & X \\
\end{tabular}
\end{center}
\label{table:tab3}
\end{table}

\section{A non-deterministic blocking method}
\label{non-determ}

\noindent
As we will see in \secref{properties}, the above deterministic blocking
method based on a greedy algorithm is sound and complete. However,
this algorithm is not optimal as it cannot block the smallest set of
potentially dangerous agents. We now introduce a non-deterministic
method, which can be used for identifying optimal solutions. The method,
which is implemented in Algorithms~\ref{non_det} and~\ref{oracolo}, is
obtained by introducing an oracle (to determine the minimum set of
agents to block so to ensure the security of the system) within the
deterministic version, which makes the soundness and completeness
results directly applicable to the non-deterministic version as well.

%



The idea is that the result given by the simulation is passed to the
method \emph{ScanOracle}, which creates all the subsets of the given set
$\mathcal{A}'$ with cardinality $|\mathcal{A}'-1|$ and finds the subsets
with the maximum number of critical formulas that remain false, using
the simulation. The simulation of all the subsets is done in parallel;
the \emph{ScanOracle} is the non-deterministic part of our algorithm.
The result is passed to the main algorithm: if we find a subset of
agents such that when executing their actions all critical formulas are
false, then we have finished and we block the remaining agents that are
not part of this subset; if not all the critical formulas remain false
the result is passed recursively to \emph{ScanOracle} until it is given
a set of agents such that all the critical formulas stay false when
simulating their actions. The rest of the agents in $\mathcal{A}'$ that
are not part of the given subset are blocked. Using this method, we can
have different best solutions but we choose one in a random way, where
with ``best solutions'' we mean sets that have the same cardinality and
are the biggest sets that make the critical formulas stay false, so that
we block the smallest set of agents that make the critical formulas
true.

\begin{figure}[t]
\begin{myalgorithm}{A non-deterministic blocking method}
\label{non_det}
\begin{algorithmic}[1]
\STATE  $\textit{Simulate}(\Gamma_n) = [\Phi' \mathcal{A}']$;
\STATE $I = \mathcal{A}'; j=0;$ 
\WHILE{($\Phi'\neq\emptyset$ \& $I\neq\emptyset$ \& $j<|\mathcal{A}'|$)}
\STATE $I' = ScanOracle(I);$
\STATE For a random $I_i \in I'$
\IF {$ |\{\phi_i \mid \not\models \phi_i$ at the current state $\}|=|\Phi'|$}
\STATE $I=\emptyset;$
\ELSE
\STATE $I = I';$
\ENDIF
\STATE $j++$
\ENDWHILE
\STATE Choose a subset $I_i\in I'$ and put $\mathcal{A}'\backslash I_i$ in $\mathcal{B}$
\end{algorithmic}
\end{myalgorithm}


\begin{myalgorithm}{ScanOracle}
\label{oracolo}
\begin{algorithmic}[1]
\STATE Generate the subset of $I$ with cardinality $|I|-1$
\STATE Execute the simulation in parallel for each subset $I_i$, where $i\in\{1, ..., |I|\}$
\STATE Take the $I_i$ with the maximum number of $\{\phi_i \mid \not\models \phi_i\}$ and put them in $I'$
\STATE $I' = I' \cup I_i$
\STATE Eliminate the duplicates in $I'$
\STATE Return $I'$
\end{algorithmic}
\end{myalgorithm}
\end{figure}


\begin{example}
As a concrete example of the application of Algorithm~\ref{non_det} (and
of Algorithm~\ref{oracolo}), consider again the system of
Example~\ref{example1}, with the same data.
%
%
The simulation of the first $4$ actions yields again $\mathcal{A}'
=~\{a_1, a_2, a_3, a_4 \}$, which is passed to ScanOracle, which in turn
creates the subsets: 
\begin{LONG}
\begin{eqnarray*}
I_1 & = & \{a_1,a_2,a_3\} \\
I_2 & = & \{a_1,a_3,a_4\} \\
I_3 & = & \{a_1,a_2,a_4\} \\
I_4 & = & \{a_2,a_3,a_4\}
\end{eqnarray*}
\end{LONG}
\begin{SHORT}
$$	
\begin{array}{ll}
I_1 = \{a_1,a_2,a_3\} & \qquad
I_2 = \{a_1,a_3,a_4\} \\
I_3 = \{a_1,a_2,a_4\} & \qquad
I_4 = \{a_2,a_3,a_4\}
\end{array}
$$
\end{SHORT}%
The oracle takes these subsets and gives as results at the current state
\begin{displaymath}
\begin{array}{lllll}
I_1: & \models\phi_1\,, & \models\phi_2\,, & \models\phi_3\,, & \models\phi_4\,. \\
I_2: & \models\phi_1\,, & \models\phi_2\,, & \models\phi_3\,, & \models\phi_4\,. \\
I_3: & \models\phi_1\,, & \not\models\phi_2\,, & \not\models\phi_3\,, & \models\phi_4\,. \\
I_4: & \not\models\phi_1\,, & \models\phi_2\,, & \models\phi_3\,, & \not\models\phi_4\,.
\end{array}
\end{displaymath}
The two subsets with maximum number of false critical formulas are $I_3$
and $I_4$, so $I'=I_3\cup I_4$. Note that, since $I_3$ and $I_4$ have
the same number of false formulas, it is enough to test just one of them
to see if all the formulas are false or not; in this case, we have just
two formulas.
The ScanOracle is then called again with $I'=I_3\cup I_4$ and it yields
the subsets
\begin{LONG}
\begin{eqnarray*}
I_5 & = & \{a_1, a_2\} \\
I_6 & = & \{a_1, a_4\} \\ 
I_7 & = & \{a_2, a_4\} \\
I_8 & = & \{a_2, a_3\} \\
I_9 & = & \{a_2, a_4\} \\
I_{10} & = & \{a_3, a_4\}
\end{eqnarray*}
\end{LONG}
\begin{SHORT}
$$	
\begin{array}{lll}
I_5 = \{a_1, a_2\} & \quad
I_6 = \{a_1, a_4\} & \quad
I_7 = \{a_2, a_4\} \\
I_8 = \{a_2, a_3\} & \quad
I_9 = \{a_2, a_4\} & \quad
I_{10} = \{a_3, a_4\}
\end{array}
$$
\end{SHORT}%
and thus the following results 
\begin{displaymath}
\begin{array}{lllll}
I_5: & \models\phi_1\,, & \not\models\phi_2\,, & \not\models\phi_3\,, & \models\phi_4\,. \\
I_6: & \not\models\phi_1\,, & \not\models\phi_2\,, & \not\models\phi_3\,, & \models\phi_4\,. \\
I_7: & \not\models\phi_1\,, & \not\models\phi_2\,, & \not\models\phi_3\,, & \not\models\phi_4\,. \\
I_8: & \not\models\phi_1\,, & \models\phi_2\,, & \models\phi_3\,, & \not\models\phi_4\,. \\
I_9: & \not\models\phi_1\,, & \not\models\phi_2\,, & \not\models\phi_3\,, & \not\models\phi_4\,. \\
I_{10}: & \not\models\phi_1\,, & \models\phi_2\,, & \models\phi_3\,, & \not\models\phi_4\,.
\end{array}
\end{displaymath}
Then $I'=I_7\cup I_9 = I_7$ as these two subsets are identical. Using
$I'$, all the critical formulas are false, so it is the maximum subset
of agents with which the system is secure. Hence, we block the remaining
agents in $\mathcal{A}'$, which is the minimum set of agents for the
blocking of which the system remains secure:
\begin{LONG}
\begin{displaymath}
\{a_1,a_2,a_3,a_4\}\backslash\{a_2,a_4\}=\{a_1,a_3\}\,.
\end{displaymath}
\end{LONG}
\begin{SHORT}
$\{a_1,a_2,a_3,a_4\}\backslash\{a_2,a_4\}=\{a_1,a_3\}$.	
\end{SHORT}
\end{example}

\section{Computational cost, Soundness and Completeness}
\label{properties}

\begin{LONG}
\noindent	
In this section, we calculate the computational cost of both the
blocking methods we have given, and then show that the greedy
deterministic method is sound and complete (which implies the same for
the non-deterministic method).
\end{LONG}
\begin{SHORT}
\noindent In this section, we list some results for the deterministic and non-deterministic methods, which are proved in~\cite{CKV10-arxiv}.
\end{SHORT}
Recall that the maximum number $n$ of actions that can be executed
in an instant of time corresponds to the cardinality of the formulas in
$\Phi$. So, in the worst case, at each instant of time, there 
are $n$ different agents that want to change the value of $n$ different
variables.

\begin{SHORT}
\begin{theorem}
The computational cost of the greedy blocking method,
Algorithm~\ref{blocking}, and of the non-deterministic blocking method,
Algorithm~\ref{non_det}, is $\mathcal{O}(n^3)$.
\end{theorem}
\end{SHORT}
\begin{LONG}
\begin{theorem}
The computational cost of the greedy blocking method,
Algorithm~\ref{blocking}, is $\mathcal{O}(n^3)$.
\end{theorem}
\begin{proof}
The simulation costs $n^2$ because if we have $n$ variables and all of
them are part of all the formulas, we need to do $n^2$ assignments. The
while cycle needs to be executed at worst $n$ times as it blocks an
agent per cycle and we may need to block all the $n$ agent. 
The cost of a single while cycle is $3n^2$ due to the sum of: the cost of the matrix creation, which is $n^2$, the cost of the association of the counter, which is $n^2$, and the cost of the quicksort algorithm, which we take to be $n^2$, instead of $n\log n$ (as the simulation has an $n^2$ complexity).
So, the total cost of the algorithm is $n^2 + (n\times (3n^2 +
n^2))=\mathcal{O}(n^3)$.
\end{proof}

\begin{theorem}
The computational cost of the non-deterministic blocking method,
Algorithm~\ref{non_det}, is $\mathcal{O}(n^3)$.
\end{theorem}
\begin{proof}
Also this algorithm uses the simulation, which costs $n^2$. In the worst
case, we need to call the $\textit{ScanOracle}$
(Algorithm~\ref{oracolo}) $n$-times, where the generation of the subsets
costs $n^2$, the oracle (that performs the simulations in parallel) is
linear and the elimination of the duplicates is $n^2$, so the cost of
the $\textit{ScanOracle}$ algorithm is $\mathcal{O}(n^2)$. The cost of
putting the agents in $\mathcal{B}$ is a constant. So, the total cost of
the algorithm is $n^2 + n\times\mathcal{O}(n^2) + C = \mathcal{O}(n^3)$.
\end{proof}

The computational cost of the non-deterministic method implemented in
Algorithm~\ref{non_det} is the same of the deterministic algorithm. This
can be puzzling for the reader, who can be expecting a lower cost, since
the non-deterministic version is obtained from the deterministic version
by using an oracle. In particular, the reduction by an oracle can be
used to prove that a problem of polynomial complexity on deterministic
machines can be solved in logarithmic time on non-deterministic
machines. We made a different choice, for specific reasons. First of
all, we employ the simulation step, that is not incorporated into the
first part of the method, and thus the solution cannot be computed in a
time lower than polynomial, in a non-deterministic fashion, even if we
use, as we did, the oracle call. As a consequence of this choice, we
could define, by using the same structure of the non-deterministic
algorithm, a variant of the algorithm in which the solutions of the
oracle are compared to each other, to choose the optimal one.
\end{LONG}

We say that a blocking method, and thus the corresponding algorithm, is
\emph{sound} if it does not allow a system to go in an insecure state
when it starts from a secure state $S_0$.
\begin{SHORT}
It is not difficult to prove that:
\end{SHORT}
\begin{theorem}
The greedy blocking method,	Algorithm~\ref{blocking}, is sound.
\end{theorem}
\begin{LONG}
\begin{proof}
For the sake of contradiction, assume that the greedy blocking method
brings a system in an insecure state. This means that the algorithm
allowed a $\Gamma_{t+1}$ to execute such that
$S_t\stackrel{\Gamma_{t+1}}{\longmapsto}S_{t+1}$ where $S_{t+1}$ is
insecure while $S_t$ is secure. But, by definition, our algorithm allows
only transitions that bring the system in a secure state, thanks to the
simulation.
\end{proof}
\end{LONG}

Let us now define the notion of a \emph{state graph}. Recall that every
state of the system is defined by an assignment of truth values to the
variables, and a state is secure if it falsifies all security-related
formulas. As a very simple example, in Figure~\ref{grafo1} we give the
state graph of a system with two variables $\{A, B\}$.

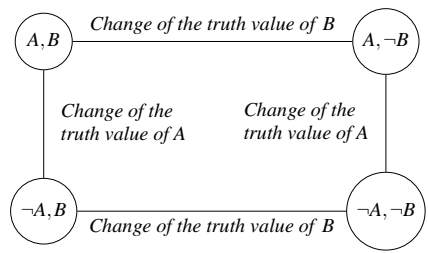
\begin{figure}[t]
\centering
\scalebox{0.75}{
\begin{tikzpicture}[scale=1.5, auto,swap]
    \node[node style sp,black] (AB) at (-2,2) {\mbox{$A,B$}};
    \node[node style sp,black] (ANB) at (2,2) {\mbox{$A,\neg{B}$}};
    \node[node style sp,black] (NAB) at (-2,0) {\mbox{$\neg{A},B$}};
	\node[node style sp,black] (NANB) at (2,0) {\mbox{$\neg{A},\neg{B}$}};

	\path (AB) edge node [above]  {\textit{Change of the truth value of} $B$} (ANB);
	\path (AB) edge node [right]  {$\begin{array}{l} \textit{Change of the} \\ \textit{truth value of } A\end{array}$} (NAB);
	\path (ANB) edge node [left]  {$\begin{array}{l} \textit{Change of the} \\ \textit{truth value of } A\end{array}$} (NANB);
	\path (NAB) edge node {\textit{Change of the truth value of} $B$} (NANB);
\end{tikzpicture}
}
\caption{A state graph for a system with two variables.}\label{grafo1}
\end{figure}

In general, to denote the transitions executed by a system, we build a
state graph as follows: every state is represented as a vertex of the
graph, and every pair of vertices is connected by an edge when and only
when the two edges differ by the truth value of one single variable,
where the edge is labeled by the name of the agent that controls that
variable. The resulting graph is \emph{indirected}. In
Figure~\ref{general}, we give an example of such a graph, where we omit
to specify all the values of the variables for readability, but instead
denote with gray vertices the insecure states and with white vertices
the secure ones.

\begin{figure}[t]
\centering
\scalebox{0.75}{
\begin{tikzpicture}[scale=0.9, auto,swap]
	\node[node style sp,black] (1) at (-2.5,1) {\mbox{$S_0$}};
        \node[node style sp,black] (2) at (-1,3) {\mbox{$S_{1,1}$}};
	\node[node style sp,black, fill=black!25] (6) at (1,0) {\mbox{$\cdots$}};
	\node[node style sp,black, fill=black!25] (3) at (-1,-1) {\mbox{$S_{1,n}$}};

	\node[node style sp,black] (4) at (1,4) {\mbox{$\cdots$}};
	\node[node style sp,black] (5) at (1,2) {\mbox{$\cdots$}};

	\node[node style sp,black] (8) at (3,3) {\mbox{$\cdots$}};
	\node[node style sp,black, fill=black!25] (b) at (2.5,1) {\mbox{$\cdots$}};
	\node[node style sp,black] (a) at (-1,1) {\mbox{$\cdots$}};

	\node[node style sp,black, fill=black!25] (7) at (1,-2) {\mbox{$\cdots$}};
	\node[node style sp,black, fill=black!25] (10) at (4,1) {\mbox{$S_{j,k}$}};
	\node[node style sp,black, fill=black!25] (9) at (3,-1) {\mbox{$\cdots$}};

	\path (1) edge node [above]  {} (2);
	\path (1) edge node [above]  {} (3);
	\path (2) edge node [above]  {} (4);
	\path (2) edge node [above]  {} (5);
	\path (3) edge node [above]  {} (6);
	\path (3) edge node [above]  {} (7);
	\path (3) edge node [above]  {} (5);
	\path (7) edge node [above]  {} (9);
	\path (6) edge node [above]  {} (9);
	\path (4) edge node [above]  {} (8);
	\path (5) edge node [above]  {} (8);
	\path (8) edge node [above]  {} (10);
	\path (9) edge node [above]  {} (10);
	\path (1) edge node [above]  {} (a);
	\path (10) edge node [above]  {} (b);
	\path (b) edge node [above]  {} (5);
	\path (a) edge node [above]  {} (6);
	\path (b) edge node [above]  {} (6);
	\path (a) edge node [above]  {} (4);
	\path (a) edge node [above]  {} (5);
	\path (b) edge node [above]  {} (4);
	\path (5) edge node [above]  {} (9);
\end{tikzpicture}
}
\caption{A state graph with secure and insecure states.}\label{general}
\end{figure}
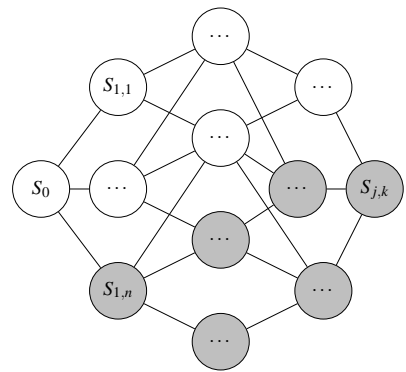

We say that a blocking method, and thus the corresponding algorithm, is
\emph{complete} if every secure state can be reached from any secure
state.
To prove the completeness of the greedy blocking method, we pursue the
following strategy:
\begin{enumerate}
\item First we prove that the state graph of the system is connected.
\item We prove that the subgraph formed by the vertices representing
secure states is connected when the security formulas can be written as a
set of Horn clauses.
\item We prove that every formula that we consider can be written as a
disjunction of Horn clauses.
\item We show that two secure states, whose security formulas can be
written as Horn clauses, are connected if and only there is a path of
secure states in the above mentioned subgraph that can be traversed by the
algorithm.
\item We show that the set of agents that have to be blocked, defined by
a rewriting into Horn clauses of a security formula, is the union of
agents that control variables occurring in one single Horn clause, and that can modify the value of the formula.
\item We show that the set of agents blocked by the greedy algorithm is
a superset of the set of agents that control variables occurring in one
single Horn clause in any rewriting of the formula.
\end{enumerate}

In particular, we write $\mathcal{H}(\Phi,\Phi')$ to denote the set of
the agents that control at least one variable of one Horn clause in one
rewriting $\Phi '$of the security formula $\Phi$, in such a way that by
changing the value of one of these variables the value of the security
formula can pass from $\bot$ to $\top$.

\begin{LONG}
Before we carry out this sequence of proof steps, let us observe a few
simple facts that will be useful in the following. 
\end{LONG}
\begin{SHORT}
Let us observe a few
simple facts that will be useful in the following. 
\end{SHORT}
First of all, every
secure state corresponds to a formula, obtained as the conjunction of
the literals representing the truth values of the variables in that
state. Since the single elements of the set of security formulas have to
be false for the system to be secure, we can describe this situation
directly by the set of secure states. Indeed, guaranteeing falseness of
each security formula corresponds to falsifying the disjunction of the
logical expressions representing the secure states.
\begin{SHORT}
By the definition of state graph, we immediately have:
\end{SHORT}
\begin{lemma}
\label{basiclemma}
The state graph is connected.
\end{lemma}
\begin{LONG}
\begin{proof}
This follows straightforwardly by the definition of state graph. 
\end{proof}
\end{LONG}

It would be tempting to presume that not only the set of states is
connected, but also the set of secure states. However, this is untrue.
Consider namely the case in which the system has two variables, $A$ and
$B$, so that there are four states as shown in Figure~\ref{grafo1}.
Suppose that the security formula is $(\neg{A} \wedge B) \vee (A
\wedge \neg{B})$. The set of secure states is formed by the
state in which both $A$ and $B$ are false and the state in which they
are both true. Clearly the set of secure states is then disconnected.

Conversely, if the set of secure states is connected, the security formula
can be written as a Horn clause (or a set of clauses, which is
equivalent). To do so, we introduce the notion of \emph{Horn rewriting}
of a formula: a propositional formula is a \emph{Horn clause} iff it can
be written in \emph{Conjunctive Normal Form} (i.e.~as a conjunction of
disjunctions of literals) in which every conjunct is formed by at most
one positive literal (This is the standard notion of Horn clause, which
we recall for preserving self-containedness.) It is well-know that every
propositional formula can be written as a disjunction of Horn clauses.

A \emph{Horn labeling} $\lambda$ of the states of a system is an
assignment of the system variables to one of the corresponding literals.
Whenever, in a Horn labeling, $\lambda(v)=v$ for a variable $v$, the
literal $v$ will be considered positive by that labeling, and
consequently literal $\neg{v}$ will be considered negative. If
$\lambda(v)=\neg{v}$, then the literal $v$ will be considered negative,
and consequently literal $\neg{v}$ will be considered positive. We
henceforth generalize the notion of Horn clause, by stating that a
formula is a Horn clause when there exists a Horn labeling for it that
makes it a Horn clause.

In the above example, the formula can be rewritten, by applying the
distributive property, as $(A \vee B) \wedge (\neg{A} \wedge\neg{B})$ and there exists a Horn labeling that makes the formula a Horn
clause: $\lambda(A)=\neg{A}$ and $\lambda(B)=B$.

\begin{LONG}
We can now prove a property that will be useful in the following.
\end{LONG}
\begin{lemma}
\label{hornlemma}
If the set of states that correspond to a security formula is
connected, then the security formula is a Horn clause.
\end{lemma}
\begin{LONG}
\begin{proof}
As said above, the security formula $\Phi$ can be written as the
disjunction of the conjunctions of literals representing the
interpretation of letters\footnote{A letter is interpreted in classical
semantics. The truth value $\top$ is represented by the positive
literal, whilst the value $\bot$ is represented by the negative
literal.} occurring in the formula itself for all the secure states. We
henceforth say, when no confusion arises, that $\Phi$ is the disjunction
of secure states.

There are two methods to obtain an equivalent formula $\Phi '$ from
$\Phi$. We can use the distributive property or we can negate the
formula obtained by considering the insecure states.
\begin{itemize}
\item For the first approach, since the variables of the system all
occur in each single state, the distribution provides two types of
subformulas in $\Phi '$, that are conjuncts of the security formula:
complete conjuncts, obtained by the combination of all variables, and
incomplete conjuncts, obtained by the combination of a subset of the
variables occurring in the set of secure states. The incomplete conjuncts
can occur only when some variables only occur in one form as a literal
in the single disjuncts of the formula $\Phi$.
\item For the second method, consider the insecure states, which form the
complement set of the secure states in the system. Clearly, the security
formula can be written as the negation of the formula obtained as a
disjunction of the insecure states. Since the negation of a disjunctive
formula is a conjunctive formula, the obtained object is in conjunctive
normal form.
\end{itemize}
The new formula is a conjunction of disjunctions of literals, in which
every literal is the disjunction of the complemented literals of an
insecure state.

We prove that, regardless of the number of involved variables, there is
always a valid Horn labeling by means of a rewriting that is, depending
on the number of involved states, either based on the distributive
property or on the negation of insecure states. In fact, given the number
of states $s$ in the set of secure states for a given formula $\Phi$ and
the number $n$ of variables of the system, we can have $s < n+1$,
$s=n+1$ or $s > n +1$.
\begin{description}

\item[$(s < n+1)$] If the number of secure states is less than or equal to
the number of variables, since this set has to be connected, by
hypothesis, then the number of complete conjuncts cannot be greater than
the number of variables. Therefore, there is always at least one
incomplete conjunct in $\Phi '$ formed by one single literal $v$. If
this happens, every time that the value of that literal is
$\neg{v}$, the security formula is false, so we can rewrite the
formula as the conjunction of the single literal conjuncts, and
consider, as Horn labeling, the one obtained by associating to the
variables in the single literal conjuncts the negation of the literals
occurring in those conjuncts.

\item[$(s = n+1)$] If the number of disjuncts in $\Phi$, namely the
number of secure states, is exactly the number $n+1$, two situations can
occur, in principle. Either there are single literal conjuncts or there
are not. Actually, due to the connectedness of the state graph, the
former is impossible. If there are single literal conjuncts, then one
variable always occurs in the same literal form in all the disjuncts of
$\Phi$. Consider one state in the set of secure states. To obtain all the
other states in the set, we should be able to change one variable at a
time and generate in this way a new state. If the system contains $k$
elements, then the above described process generates $k-1$ new states.
But if there is one single literal conjunct in the rewriting $\Phi '$,
then at least one variable cannot change ever during this process.
Therefore, we can only generate $n-1$ new states contrary to the
hypothesis. When a set of states is formed in the way described above,
then every single state always contains the same number of positive and
negative literals (with respect to the trivial Horn labeling), apart
from one that contains one positive or one negative literal more than
the others. This special state can be used to generate the correct Horn
labeling for this set by inverting each literal in the $\lambda$
function (positive literals make the variables corresponding to negative
literals and vice versa).

\item[$(s > n+1)$] If the number of secure states is greater than $n+1$,
then we can consider the rewriting of $\Phi$ obtained by negating the
disjunction of insecure states, that we call $\Phi_U$. In other words,
$\Phi = \neg \Phi_U$. If we distribute the literals within $\Phi_U$
and obtain the conjunctive normal form of $\Phi_U$, then applying the
negation and again distributing the literals, we finally have a
conjunctive normal form for $\Phi$ that, based on the properties of the
graph, enjoys the same property of existence of single literal conjuncts
discussed for the case $s < n+1$. Therefore the claim is proved.
\end{description}
\end{proof}
\end{LONG}

Let us now consider a generic set of states that is not connected. As we
show in Figure~\ref{grafo1}, this may anyhow correspond to a valid Horn
labeling. This, however, does not occur for every security formula.
Conversely, every set of states can be written as the intersection of
connected sets of states. Therefore, given any security formula, we can
represent it as the disjunction of the Horn clauses that are obtained by
the sets of connected states.

The purpose of Algorithm~\ref{blocking} is to block the agents that
apply for changing variables so to make true a critical formula. Since a
critical formula can be made true by making true one of its disjuncts,
Lemma~\ref{hornlemma} can be used directly to prove the following
Lemma~\ref{ffprime}.

More specifically, the greedy blocking works by blocking agents when
they apply for the modification of the truth value of a variable, where
the blocking condition is: an agent cannot perform an action when this
performance brings the system in an insecure state. The synchronization
proposed by the algorithm is based on application time: the system
simulates the result of performing all (up to the maximum) actions that
agents applied for at that instant of time. The system denies the
execution to those agents that modify variables involved in the
transition of the system into an insecure state. Since this may
correspond to more than one combination, the resulting blocked agent set
may be larger than needed. We can assume, therefore, without loss of
generality, that the algorithm blocks all the agents that applied for
modifying variables that bring the system into an insecure state. This
assumption is sufficient to employ fruitfully the generalization of
Lemma \ref{hornlemma} to generic formulas. Remember that
$\mathcal{H}(\Phi,\Phi')$ denotes the set of the agents that control
variables of one Horn clause in the rewriting $\Phi '$of $\Phi$ and
bring the system into an insecure state.

\begin{lemma}\label{ffprime}
If no agent in $\mathcal{H}(\Phi,\Phi ')$ modifies variables occurring in $\Phi$, and $\Phi$ is false, then $\Phi$ is false after the modifications.
\end{lemma}
\begin{LONG}
\begin{proof}
This is a direct consequence of the proof of Lemma~\ref{hornlemma}. \
\end{proof}
\end{LONG}



Since the agents blocked by the algorithm are all those that bring the
system into an insecure state, then every agent controls variables that
certainly occur in at least one disjunct of $\Phi$. If we rewrite $\Phi$
as a disjunction of Horn clauses (following the standard notion of
formula rewriting into disjunction of Horn clauses)
\begin{LONG}
\begin{displaymath}
\Phi' = \Phi'_1 \vee \Phi'_2 \vee ... \vee \Phi'_k\,,
\end{displaymath}
\end{LONG}%
\begin{SHORT}
$\Phi' = \Phi'_1 \vee \Phi'_2 \vee ... \vee \Phi'_k$,
\end{SHORT}%
then, by definition of this rewriting, every variable controlled by $a$
that occurs in $\Phi$, occurs in at least one of these disjuncts. If an
agent that controls one variable is blocked by our algorithm, then, by
definition of the simulation, at least one of the conjuncts in which the
variable occurs in $\Phi '$ is true. This means that given any pair of
secure states $s$ and $s'$, the algorithm never blocks an agent that
brings the system directly from $s$ to $s'$. The extension of this
property to paths is proved in the following theorem.
\begin{theorem}
The greedy blocking method,	Algorithm~\ref{blocking}, is complete.
\end{theorem}
\begin{LONG}
\begin{proof}
Consider two secure states $s$ and $s'$, and one agent $a$ that controls
variables occurring in $\Phi$. Suppose that $s$ and $s'$ are connected
by a path of length $k$. If we can connect $s$ to $s'$ by a path of
length $k$, then we can connect $s$ to $s''$ by a path of length $k-1$
and directly $s''$ to $s'$. Suppose, by contradiction, that the greedy
method blocks too many agents, so that the system would be able to move
from $s''$ to $s'$ but not from $s$ to $s''$. By Lemma~\ref{ffprime}, if one agent that is not blocked by the algorithm does not apply for changing one variable that occurs in one of the Horn clauses that appear in a rewriting of $\Phi$, then the formula remains false. 
This would mean that at
least one agent, able to perform the changes of variable values that
would bring the system from $s$ to $s'$, controls at least one variable
that does not occur in each possible rewriting of $\Phi$. 
This leads to
the conclusion that at least one variable is controlled by one agent
$a$, blocked by the greedy method, since $a$ is able to bring the system
into an insecure state, from $s''$ to $s'$, the same variable is
controlled by another agent $a'$, that vice versa could move the system
to an insecure state and occurs in one Horn disjunct of one of the
possible rewritings of $\Phi$, and brings the system from $s$ to $s''$.
As a consequence, one variable would be controlled by two agents, which
is contrary to the definition of agent control of variables in the
system definition adopted here.
\end{proof}
\end{LONG}

\noindent
Soundness and completeness of the deterministic algorithm directly extend to the non-deterministic one.
\begin{LONG}
\begin{theorem}
\label{theoqui}
The non-deterministic blocking method, Algorithm~\ref{non_det}, is sound.
\end{theorem}
\begin{proof}
The non-deterministic blocking method is an extension of the greedy blocking method of Algorithm~\ref{blocking} by means of an oracle. This means that every solution of Algorithm~\ref{blocking} is also a solution of Algorithm~\ref{non_det}.
\end{proof}
\begin{theorem}
The non-deterministic blocking method, Algorithm~\ref{non_det}, is complete.
\end{theorem}
\begin{proof}
The same reasoning used to prove Theorem~\ref{theoqui} applies for completeness.
\end{proof}
\end{LONG}
\begin{SHORT}
\begin{theorem}
The non-deterministic blocking method, Algorithm~\ref{non_det}, is sound and complete.
\end{theorem}
\end{SHORT}

Let us call \emph{optimal} a method that blocks the smallest sets of
agents to ensure the security of the system. The greedy blocking method
guarantees just one of the optimality properties, i.e.~security, but it
cannot guarantee to block the smallest sets of agents. We thus say that
the greedy blocking method is a \emph{sub-optimal solution}.
\begin{LONG}

\end{LONG}%
What can further be proved is that the comparison of the solutions
computed in the non-deterministic method generates an optimal solution.
This is quite obvious, since the solutions computed are all the possible
combinations, and thus the best solution is included in this set. What
the algorithm does is find the smallest set of agents that need to be
blocked.\footnote{There may exist more than one solution with the
smallest number of agents blocked. The approach of Algorithm
\ref{non_det} is to compare everything with everything, so the chosen
solution is the last examined one.}

\begin{theorem}
Algorithm \ref{non_det} computes an optimal solution.
\end{theorem}

We consider here the specific problem of blocking underhand attacks as
the problem of keeping the security formula false when agents apply for
changing variables.
The computational complexity of a problem is defined as the cost of the
best solution. In this case, we cannot claim that the solution is
optimal and therefore we only have an upper bound result.

\begin{theorem}
Blockage of underhand attacks is a polynomially solvable problem on deterministic machines.
\end{theorem}
\begin{LONG}
\begin{proof}
Algorithm \ref{blocking} is deterministic, sound and complete, and its cost is polynomial.
\end{proof}
\end{LONG}

Analogously, the next result is a consequence of the results about
soundness, completeness and cost of Algorithm~\ref{non_det}, again in
form of an upper bound.

\begin{theorem}
Optimal blockage of underhand attacks is a polynomially solvable problem on non-deterministic machines.
\end{theorem}

\section{Conclusions}
\label{concl}


\begin{LONG}
\noindent	
In this work, we have dealt with multi-agent systems whose security
properties depend on the values of sets of logical formulas (the
critical formulas of the systems). We assumed that the values of these
formulas can change due to actions performed by the agents of the
system, and that some attacks can be performed by malicious agents that
are authorized within the system itself (in other terms, users of the
system). These attacks conducted from inside are underhand, and we focus
specifically on those attacks that are performed by groups of
individuals that do not reveal their belonging to such groups, that we
call hidden coalitions.

We have introduced a deterministic method, implemented by the greedy
blocking algorithm, which prevents attacks to the system carried out by
hidden coalitions formed by agents that are users of the system itself.
The method based on this algorithm is sound and complete,
but the algorithm is not optimal as it cannot block the smallest set of
potentially dangerous agents.
%
%
The method is thus extended to a non-deterministic version that can be
used, in future investigations, to identify optimal solutions and to
study extensively the computational properties of the solution, from
both deterministic and non-deterministic sides. The method is obtained
by introducing an oracle within the deterministic version, which makes
the soundness and completeness results directly applicable to the
non-deterministic version as well.

The starting point of our approach to model multi-agent systems is 
Coalition Logic~\cite{pauly,pauly2,pauly3}, a cooperation
Logic that implements ideas of Game Theory. Another cooperation logic
that works with coalitions is the Alternating-time Temporal
Logic~\cite{atl3,atl,atl2}. A widely used logic, specifically thought
for dealing with strategies and multi-agent systems, is the Quantified
Coalition Logic~\cite{quantified2,quantified,quantified3}. A specific
extension, also used for agents in multi-agent systems is
CL-PC~\cite{cl-pc,cl-pc3}, and this is indeed the version of Coalition
Logic that we started from.
\end{LONG}

\begin{SHORT}
\noindent	
We have given a first approach to hidden coalitions by introducing a
deterministic method that blocks the actions of potentially dangerous
agents, i.e.~possibly belonging to such coalitions. We have also given a
non-deterministic version of this method that blocks the smallest set of
potentially dangerous agents. Our two blocking methods are sound and
complete, and we have calculated their computational cost.

The starting point of our approach to model multi-agent systems is
Coalition Logic~\cite{pauly,pauly3}, a cooperation Logic that implements
ideas of Game Theory. Another cooperation logic that works with
coalitions is the Alternating-time Temporal Logic, e.g.~\cite{atl}. A
widely used logic, specifically thought for dealing with strategies and
multi-agent systems, is the Quantified Coalition Logic~\cite{quantified}.
A specific extension, also used for agents in multi-agent systems is
CL-PC~\cite{cl-pc,cl-pc3}, and this is indeed the version of Coalition
Logic that we started from.
\end{SHORT}

The notion of hidden coalition is a novelty, and more generally, to the
best of our knowledge, no specific investigation exists that deals with
security in open systems by means of a notion of underhand attack. The
system presented here is a multi-agent one, where we did not discuss how
these coalitions are formed or the negotiations that can take place
before the creation of the coalitions~\cite{negotiation,negotiation2}.
For future work, it will be interesting to consider in more detail the
non-monotonic aspects underlying the problem of underhand attacks by
hidden coalitions, e.g.~to formalize: the mental attitudes and
properties of the intelligent agents that compose the system, how agents
enter/exit/are banned from a system or enter/exit a hidden coalition,
and the negotiations between the agents for establishing the common goal
and synchronizing/organizing their actions. In this work, we give a way
to protect the system, without making a distinction between the case in
which the agents that make the attack are actual members of a coalition
or not. If the system is equipped with explicit/implicit coalition test
methods, this can make up a significant difference in terms of
usefulness of our approach.


A specific analysis of the computational properties of our blocking
methods, in particular an analysis of worst, average, and practical
cases, will be subject of future work. Results of lower bound for the
blocking problem and the optimal blocking problem, and the computational
cost of the optimal blocking problem on deterministic machines are in
particular important aspects to be investigated.


\section*{ACKNOWLEDGEMENTS}
\noindent	
The work presented in this paper was partially supported by the
FP7-ICT-2007-1 Project no.~216471, ``AVANTSSAR: Automated Validation of
Trust and Security of Service-oriented Architectures''.

\renewcommand{\baselinestretch}{0.98}
\bibliographystyle{apalike}
{\small
\bibliography{erisa}}
\renewcommand{\baselinestretch}{1}

\end{document}